\documentclass[conference]{IEEEtran}
\usepackage{bm}
\usepackage{cite}
\usepackage{amsmath,amssymb,amsfonts}
\usepackage{algorithm}
\usepackage{algpseudocode}
\usepackage{graphicx}
\usepackage{textcomp}
\usepackage{xcolor}
\usepackage{subfigure}
\allowdisplaybreaks[4]

\def\BibTeX{\rm B\kern-.05em{\sc i\kern-.025em b}\kern-.08emT\kern-.1667em\lower.7ex\hbox{E}\kern-.125emX}

\newtheorem{lemma}{Lemma}

\begin{document}
\title{MiLAC-Aided Beamforming for MIMO Over-the-Air Computation}
\author{\IEEEauthorblockN{Yaru Wang, Deyou Zhang, Qingchao Li, Jun Liu, and Chuang Shi}
\IEEEauthorblockA{\textit{School of Electronic and Information Engineering, Beihang University, Beijing 100191, China} \\
email: \{yaruw, deyou, qingchaoli, liujun2019, shichuang\}@buaa.edu.cn}
}

\maketitle

\begin{abstract}
Over-the-air computation (AirComp) enables low-latency wireless data aggregation, but its accuracy is limited by imperfect signal alignment over fading channels and receiver noise. Fully digital beamforming improves aggregation accuracy in multiple-input multiple-output (MIMO) AirComp systems but requires one radio-frequency (RF) chain per antenna. To reduce this hardware burden, we investigate microwave linear analog computer (MiLAC)-aided beamforming for MIMO AirComp. Under a lossless and reciprocal MiLAC model, we jointly optimize the transmit digital precoding matrices and the receive-side MiLAC aggregation matrix to minimize the mean squared error (MSE). An alternating optimization algorithm is developed, in which the precoding matrices are optimally updated using the Karush--Kuhn--Tucker conditions and bisection, while the resulting convex aggregation matrix subproblem is solved globally using projected gradient descent. Numerical results verify the algorithm's convergence and demonstrate that MiLAC-aided beamforming approaches the MSE performance of fully digital beamforming with substantially fewer RF chains and outperforms phase-shifter-based hybrid beamforming under the same RF-chain budget.
\end{abstract}

\begin{IEEEkeywords}
Microwave linear analog computer, beamforming, multiple-input multiple-output, over-the-air computation.
\end{IEEEkeywords}

\IEEEpeerreviewmaketitle

\section{Introduction}
The Internet of Things (IoT) is expected to connect billions of edge devices (EDs), making wireless data aggregation (WDA) essential to many emerging applications \cite{FudongLi2022Utility-ICC}. However, conventional WDA separates communication from computation and becomes increasingly inefficient as the number of EDs grows under limited spectrum and stringent latency constraints \cite{GuangxuZhu2021OTA-WC}. Over-the-air computation (AirComp) addresses this limitation by exploiting waveform superposition to aggregate concurrently transmitted data, thereby enabling low-latency and spectrum-efficient WDA \cite{GuangxuZhu2021OTA-WC, AlphanSahin2023Survey-CST, ZhengChen2023OTA-Network}.

The performance of AirComp relies on signal alignment among distributed EDs, which is challenging over fading channels, while receiver noise further degrades aggregation accuracy. Multiple-input multiple-output (MIMO) techniques have therefore been incorporated into AirComp systems to exploit additional spatial degrees of freedom (DoFs), motivating extensive studies on beamforming for MIMO AirComp \cite{GuangxuZhu2019MIMO-IOTJ, DingzhuWen2019Reduced-TWC, XuShi2024Beamforming-TWC}.

Implementing MIMO AirComp over large-scale antenna arrays using fully digital beamforming incurs high hardware cost and power consumption because each antenna requires a dedicated radio-frequency (RF) chain. Phase-shifter-based hybrid beamforming reduces the number of RF chains \cite{XiongfeiZhai2021Hybrid-TCOM, ShusenJing2023Transceiver-TWC}, but its constant-modulus constraint limits the flexibility of analog-domain transformations for multi-stream signal alignment. This limitation motivates alternative analog architectures with greater processing flexibility.

The microwave linear analog computer (MiLAC) is a reconfigurable microwave network capable of implementing a broad class of linear transformations directly in the analog domain \cite{MatteoNerini2025MiLACI-TSP, MatteoNerini2025MiLACII-TSP}. MiLAC-aided beamforming has recently been investigated for large-scale MIMO systems and can achieve the capacity of fully digital beamforming under lossless and reciprocal constraints \cite{MatteoNerini2025MiLACII-TSP, MatteoNerini2026MIMO-TWC, MatteoNerini2025Capacity-arXiv}. However, existing designs focus on communication-oriented objectives, such as rate and capacity maximization, rather than the computation-oriented signal alignment required by AirComp. To the best of our knowledge, MiLAC-aided beamforming for MIMO AirComp remains unexplored.

Motivated by this gap, we investigate a MiLAC-aided MIMO AirComp system comprising multiple multi-antenna EDs and a MiLAC-equipped multi-antenna AP. The main contributions are summarized as follows.

1) Using an equivalent characterization of lossless and reciprocal MiLACs, we express the physical feasibility of the MiLAC aggregation matrix through a spectral-norm constraint and formulate a joint MSE minimization problem over the ED-side digital precoding matrices and the AP-side MiLAC aggregation matrix.

2) We develop an alternating optimization (AO) algorithm in which the digital precoding matrices are optimally updated using the Karush--Kuhn--Tucker conditions and one-dimensional bisection, while the convex MiLAC aggregation matrix subproblem is solved to its global optimum using projected gradient descent (PGD). The optimized aggregation matrix is subsequently mapped to the MiLAC scattering and admittance matrices.

3) Numerical results show that the proposed design approaches the MSE performance of fully digital beamforming with substantially fewer RF chains at the AP and outperforms phase-shifter-based hybrid beamforming under the same RF-chain budget.

\section{System Model and Problem Formulation}\label{sec:SMPF}
\subsection{System Model}
As illustrated in Fig.~\ref{fig:SystemModel}, we consider a MiLAC-aided MIMO AirComp system consisting of an $M$-antenna AP and $K$ distributed EDs, each equipped with $N$ antennas, where $M \gg N$. Each ED $k \in \mathcal K \triangleq \{1, 2, \cdots, K\}$ employs $N$ RF chains to perform fully digital beamforming and transmit an $L$-dimensional symbol vector, where $L \le N$. At the AP, an $(M + L)$-port MiLAC is cascaded with $L$ RF chains to perform signal aggregation. Following \cite{MatteoNerini2025Capacity-arXiv}, the RF chains at each ED are modeled as voltage generators with a series impedance $Z_0$, e.g., $Z_0 = 50~\Omega$, whereas those at the AP are modeled as terminals loaded with the same impedance. Throughout this paper, all RF-domain voltage signals are represented by normalized complex envelopes referenced to $Z_0$, such that their squared Euclidean norms represent normalized signal powers. The reference impedance $Z_0$ is retained for characterizing the circuit-level input--output relationship of the MiLAC.

Let $\bm s_k = [s_{k, 1}, \cdots, s_{k, L}]^T \in \mathbb C^{L \times 1}$ denote the symbol vector transmitted by ED $k$, $\forall k \in \mathcal{K}$. For simplicity, the symbol vectors $\{\bm s_k\}$ are assumed to be independent and identically distributed (i.i.d.), with $\bm s_k \sim \mathcal {CN} (\bm 0, \bm I_L)$. The target function vector to be recovered at the AP is the arithmetic mean of all EDs' symbol vectors, given by
\begin{equation}
    \bm s = \frac{1}{K} \sum_{k=1}^K \bm s_k.
\end{equation}

For ED $k$, let $\bm W_k \in \mathbb C^{N \times L}$ denote its digital precoding matrix. The normalized source signal generated by its RF chains can then be expressed as
\begin{equation}
    \bm c_k = \bm W_k \bm s_k, ~\forall k \in \mathcal K.
\end{equation}
Since $\mathbb{E}[\bm s_k \bm s_k^H] = \bm I_L$, the corresponding source-signal power is $\mathbb{E}[\|\bm c_k\|_2^2] = \|\bm W_k\|_F^2$. Let $P_k$ denote the maximum normalized source-signal power budget of ED $k$. The digital precoding matrix therefore satisfies $\|\bm W_k\|_F^2 \le P_k$. Following \cite{MatteoNerini2026MIMO-TWC}, under perfect impedance matching between the transmit antennas and the source impedance $Z_0$, the signal delivered to the transmit antennas is $\bm x_k = \bm c_k/2$.

Based on the AirComp mechanism, all EDs simultaneously transmit their signals over the same time--frequency resource. The signal received at the AP is then given by
\begin{equation}
    \bm y = \sum_{k = 1}^K \bm H_k \bm x_k + \bm n = \frac{1}{2} \sum_{k = 1}^K \bm H_k \bm W_k \bm s_k + \bm n,
\end{equation}
where $\bm H_k \in \mathbb C^{M \times N}$ denotes the channel from ED $k$ to the AP, and $\bm n \sim \mathcal {CN} (\bm 0, \sigma_n^2 \bm I_M)$ denotes the additive white Gaussian noise (AWGN) vector at the AP. The received signal $\bm y$ is subsequently processed by the MiLAC, yielding the estimated function vector
\begin{equation}
    \begin{aligned}
        \hat{\bm s} = & \frac{1}{K} \bm F_{\rm MiLAC} \bm y \\[1ex]
        = & \frac{1}{K} \left(\frac{1}{4}\sum_{k = 1}^K \bm F \bm H_k \bm W_k \bm s_k + \frac{1}{2} \bm F \bm n \right),
    \end{aligned}
\end{equation}
where $\bm F_{\rm MiLAC} \in \mathbb C^{L \times M}$ is the aggregation matrix physically implemented by the MiLAC, and $\bm F \triangleq 2 \bm F_{\rm MiLAC}$ is its scaled counterpart introduced for notational convenience.

\begin{figure}[t!]
\vskip2pt
\centering
\includegraphics[width = 6.8cm]{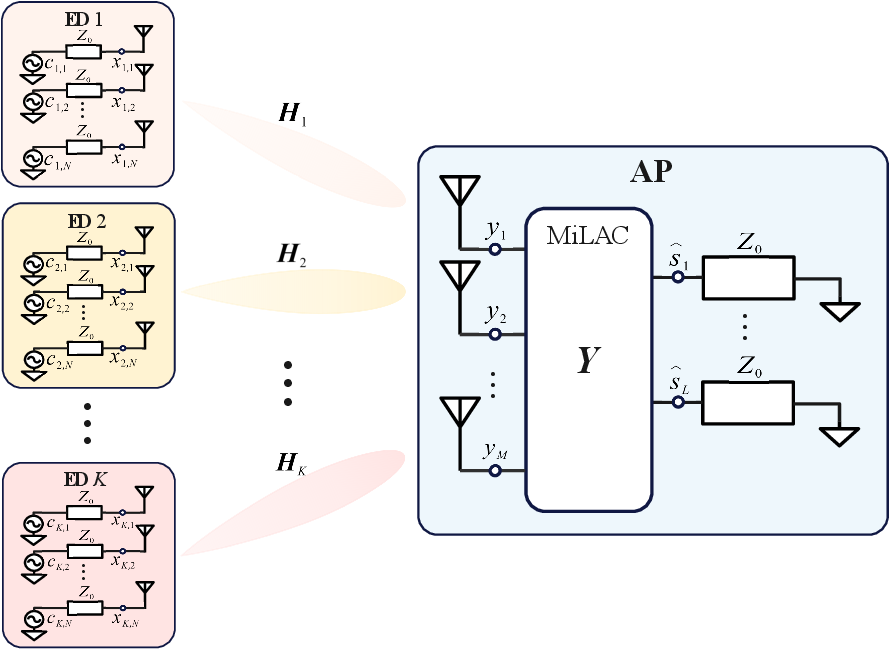}
\caption{A MiLAC-aided MIMO AirComp system.} \label{fig:SystemModel}
\end{figure}

\subsection{MiLAC Modeling}
In this paper, we consider the fully connected MiLAC topology proposed in \cite{MatteoNerini2025MiLACI-TSP, MatteoNerini2025MiLACII-TSP}, in which each port is connected to ground through a tunable shunt admittance, while each pair of distinct ports is interconnected through a tunable mutual admittance. We index the first $M$ ports as the input ports connected to the AP antennas and the remaining $L$ ports as the output ports connected to the RF chains. Let $\bm Y \in \mathbb C^{(M+L) \times (M+L)}$ denote the admittance matrix of the resulting $(M+L)$-port microwave network.

Following \cite{MatteoNerini2025MiLACI-TSP, MatteoNerini2025MiLACII-TSP}, we focus on a lossless and reciprocal MiLAC. According to microwave network theory \cite{Pozar2011microwaveengineering}, losslessness and reciprocity require the admittance matrix to be skew-Hermitian and symmetric, respectively, i.e., $\bm Y^H = - \bm Y$ and $\bm Y^T = \bm Y$. These conditions imply that $\bm Y$ is a purely imaginary and symmetric matrix. Moreover, since the eigenvalues of a skew-Hermitian matrix are purely imaginary, the matrix $\bm I_{M+L} + Z_0 \bm Y$ is nonsingular for $Z_0 > 0$.

As established in \cite{MatteoNerini2025MiLACI-TSP, MatteoNerini2025MiLACII-TSP}, the linear transformation from the $M$ input ports to the $L$ output ports is characterized by the aggregation matrix
\begin{equation}\label{eq:F_MiLAC}
    \bm F_{\rm MiLAC} = \left[ (\bm I_{M + L} + Z_0 \bm Y)^{-1} \right]_{M + 1 : M + L, 1 : M},
\end{equation}
where $[\bm A]_{\mathcal I, \mathcal J}$ denotes the submatrix of $\bm A$ formed by the rows indexed by $\mathcal I$ and the columns indexed by $\mathcal J$.

Although $\bm Y$ directly characterizes the circuit parameters of the MiLAC, the dependence of $\bm F_{\rm MiLAC}$ on $\bm Y$ through a matrix inverse makes it difficult to directly characterize the feasible set of $\bm F_{\rm MiLAC}$. We therefore introduce the scattering matrix $\bm \Theta \in \mathbb C^{(M+L) \times (M+L)}$ as an equivalent representation of the MiLAC. With a common reference impedance $Z_0$ at all ports, the scattering and admittance matrices are related through the Cayley transform \cite{Pozar2011microwaveengineering}
\begin{subequations}\label{eq:Theta_Y_relation}
\begin{align}
\bm \Theta & = \left(\bm I_{M+L}+Z_0\bm Y\right)^{-1} \left(\bm I_{M+L}-Z_0\bm Y\right) \label{eq:Theta_Y_relation_a} \\[1ex]
& \overset{(a)}{=} -\bm I_{M+L} + 2\left(\bm I_{M+L}+Z_0\bm Y\right)^{-1}, \label{eq:Theta_Y_relation_b}
\end{align}
\end{subequations}
where $(a)$ follows by writing $\bm I_{M+L} - Z_0 \bm Y = 2 \bm I_{M+L} - \left(\bm I_{M+L} + Z_0 \bm Y\right)$. Since the submatrix of $\bm I_{M+L}$ corresponding to rows $(M+1 : M+L)$ and columns $(1 : M)$ is a zero matrix, \eqref{eq:F_MiLAC} and \eqref{eq:Theta_Y_relation_b} yield $\bm F_{\rm MiLAC} = \frac{1}{2}[\bm \Theta]_{M+1:M+L,1:M}$ or equivalently
\begin{equation}\label{eq:F_MiLAC_Theta}
    \bm F = [\bm \Theta]_{M+1:M+L,1:M}.
\end{equation}

Furthermore, the skew-Hermitian property of $\bm Y$ ensures that its Cayley transform is unitary, whereas the symmetry of $\bm Y$ ensures that $\bm\Theta$ is symmetric. Therefore,
\begin{equation}\label{eq:Theta_constraints}
\bm \Theta^H \bm \Theta = \bm I_{M+L},~
\bm \Theta^T = \bm\Theta.
\end{equation}
Consequently, the scaled aggregation matrix $\bm F$ is the input--output block of a symmetric and unitary scattering matrix. This characterization is used in the following subsection to express the feasible set of $\bm F$ through a spectral-norm constraint.

\subsection{Problem Formulation}
Our objective is to minimize the distortion between $\hat{\bm s}$ and $\bm s$, as measured by the MSE defined as
\begin{align}
    {\rm MSE} = &~ \mathbb E[\|\hat{\bm s} - \bm s\|_2^2] \nonumber \\[2ex]
    = &~ \frac{1}{K^2} \left(\mathbb{E} \left[ \left\|\sum\limits_{k = 1}^K \bm Z_k \bm s_k + \frac{1}{2} \bm F \bm n \right\|_2^2 \right]\right) \nonumber \\[2ex]
    \overset{(a)}{=} &~ \frac{1}{K^2}\left(\sum\limits_{k = 1}^K \left\| \bm Z_k \right\|_F^2 + \frac{1}{4} \sigma_n^2 \left\|\bm F \right\|_F^2\right), \label{eq:MSE}
\end{align}
where $\bm Z_k = \frac{1}{4} \bm F \bm H_k \bm W_k - \bm I_L$, and $(a)$ follows from the mutual independence of $\bm s_1, \ldots, \bm s_K$, and $\bm n$.

To this end, we formulate the following optimization problem:
\begin{subequations}\label{OP1}
    \begin{align}
        \min_{\{\bm W_k\}, \bm F, \bm \Theta} &~~ {\rm MSE} \\
        {\rm s.t.} &~~ \left\|\bm W_k\right\|_F^2 \le P_k, \forall k \in \mathcal K, \\[1ex]
        &~~ \bm F = [\bm \Theta]_{M + 1 : M + L, 1 : M}, \\[1ex]
        &~~\bm \Theta^H \bm \Theta = \bm I_{M+L},~
\bm \Theta^T = \bm\Theta.
    \end{align}
\end{subequations}

The full scattering matrix $\bm \Theta$ does not need to be optimized explicitly. In particular, \cite[Proposition 1]{ZheyuWu2026MiLACMISO-arXiv} establishes that the existence of a symmetric and unitary scattering matrix having $\bm F$ as its input--output block is equivalent to a spectral-norm constraint on $\bm F$. We restate this result in the following lemma.

\begin{lemma}\label{lem:unitary_completion}
Let $\bm \Theta \in \mathbb C^{(M + L) \times (M + L)}$ be a symmetric and unitary matrix, and define $\bm F = [\bm \Theta]_{M + 1 : M + L, 1 : M}$. Then, $\|\bm F\|_2 \le 1$. Conversely, for any matrix $\bm F \in \mathbb C^{L \times M}$ satisfying $\|\bm F\|_2 \le 1$, there exist symmetric matrices $\bm \Theta_{11} \in \mathbb C^{M \times M}$ and $\bm \Theta_{22} \in \mathbb C^{L \times L}$ such that
\begin{equation*}
    \bm \Theta = \left[ \begin{array}{cc}\bm \Theta_{11} & \bm F^T \\\bm F & \bm \Theta_{22}\end{array} \right]
\end{equation*}
is unitary and symmetric.
\end{lemma}

\begin{IEEEproof}
This result follows from \cite[Proposition 1]{ZheyuWu2026MiLACMISO-arXiv}.
\end{IEEEproof}

By Lemma~\ref{lem:unitary_completion}, the scattering matrix $\bm\Theta$ can be eliminated from \eqref{OP1}, yielding the following equivalent problem:
\begin{subequations}\label{eq:prob_MiLAC_simplification}
    \begin{align}
        \min_{\{\bm W_k\}, \bm F} &~~ {\rm MSE} \\
        {\rm s.t.} &~~ \left\|\bm W_k\right\|_F^2 \le P_k, \forall k \in \mathcal K, \\[2ex]
        &~~ \left\|\bm F\right\|_2 \le 1.
    \end{align}
\end{subequations}

\section{Alternating Optimization for Beamforming Design}\label{Sec:AO}
Due to the coupled terms $\bm F \bm H_k \bm W_k$, problem \eqref{eq:prob_MiLAC_simplification} is not jointly convex in $\bm F$ and $\{\bm W_k\}$. To address this coupling, we develop an AO framework that alternates between optimizing $\{\bm W_k\}$ for a fixed $\bm F$ and optimizing $\bm F$ for fixed $\{\bm W_k\}$. These two blocks are updated iteratively until a prescribed convergence criterion is satisfied.

\subsection{Optimization of $\{\bm W_k\}$}
For a fixed $\bm F$, and after dropping the terms independent of $\{\bm W_k\}$ from the MSE in \eqref{eq:MSE}, problem \eqref{eq:prob_MiLAC_simplification} reduces to
\begin{subequations}\label{eq:opt_W}
    \begin{align}
        \min_{\{\bm W_k\}} ~& \sum_{k = 1}^K \left( {\rm Tr} ( \bm W_k^H \bm A_k^H \bm A_k \bm W_k ) - 2 \Re \left\{ {\rm Tr} (\bm A_k \bm W_k) \right\} \right) \label{eq:opt_W_obj} \\[1ex]
        {\rm s.t.} ~&~ {\rm Tr} \left( \bm W_k^H \bm W_k \right) \le P_k, \forall k \in \mathcal K,
    \end{align}
\end{subequations}
where $\bm A_k \triangleq \frac{1}{4} \bm F \bm H_k \in \mathbb C^{L\times N}$.

Since the precoding matrices $\{\bm W_k\}$ in \eqref{eq:opt_W} are decoupled across EDs, problem \eqref{eq:opt_W} can be decomposed into $K$ independent subproblems. Specifically, for ED $k$, the optimal $\bm W_k$ can be obtained by solving
\begin{subequations}\label{eq:opt_Wk}
    \begin{align}
        \min_{\bm W_k}~~& {\rm Tr} (\bm W_k^H \bm A_k^H \bm A_k \bm W_k) - 2 \Re \left\{ {\rm Tr} (\bm A_k \bm W_k) \right\} \\[1ex]
        {\rm s.t.}~~& {\rm Tr} \left( \bm W_k^H \bm W_k \right) \le P_k. \label{eq:opt_Wk_const}
    \end{align}
\end{subequations}
Problem \eqref{eq:opt_Wk} is a convex quadratically constrained quadratic program (QCQP) and Slater's condition holds. Therefore, we can employ the KKT conditions to solve \eqref{eq:opt_Wk} optimally.

Let $\lambda_k \ge 0$ denote the Lagrange multiplier associated with the constraint \eqref{eq:opt_Wk_const}. Then, the Lagrangian of \eqref{eq:opt_Wk} is given by
\begin{align}
    \mathcal{L}_k = & {\rm Tr} ( \bm W_k^H \bm A_k^H \bm A_k \bm W_k ) - 2 \Re \left\{ {\rm Tr} (\bm A_k \bm W_k) \right\} \nonumber \\[2ex]
    & + \lambda_k \left({\rm Tr} \left( \bm W_k^H \bm W_k \right) - P_k\right).
\end{align}
The corresponding KKT conditions are expressed as
\begin{subequations}\label{eq:opt_Wk_KKT}
    \begin{align}
        \frac{\partial \mathcal{L}_k}{\partial \bm W_k^*} = \bm A_k^H \bm A_k \bm W_k - \bm A_k^H + \lambda_k \bm W_k = \bm 0_{N \times L}, \label{eq:opt_Wk_KKT_stationarity} \\
        {\rm Tr} \left( \bm W_k^H \bm W_k \right) \le P_k, \label{eq:opt_Wk_KKT_primalfeasibility} \\[1ex]
        \lambda_k \ge 0, \\[1ex]
        \lambda_k \left( {\rm Tr} \left( \bm W_k^H \bm W_k \right) - P_k \right) = 0. \label{eq:opt_Wk_KKT_complementary}
    \end{align}
\end{subequations}
Using the equality
\begin{equation*}
\left(\bm A_k^H \bm A_k + \lambda_k \bm I_N \right) \bm A_k^H = \bm A_k^H \left( \bm A_k \bm A_k^H + \lambda_k \bm I_L \right),
\end{equation*}
a solution to the stationarity condition \eqref{eq:opt_Wk_KKT_stationarity} is
\begin{equation}\label{eq:opt_Wk_star}
\bm W_k(\lambda_k) = \bm A_k^H \left( \bm A_k \bm A_k^H + \lambda_k \bm I_L \right)^{-1}.
\end{equation}
Under the considered continuous fading channel model and $L \le N$, $\bm A_k$ is assumed to have full row rank, i.e., $\operatorname{rank}(\bm A_k) = L$, which holds almost surely for a nondegenerate full row rank $\bm F$. Therefore, \eqref{eq:opt_Wk_star} is well defined for every $\lambda_k \ge 0$ including $\lambda_k = 0$.

The optimal multiplier $\lambda_k^\star$ is determined as follows. If $\|\bm W_k(0)\|_F^2 \le P_k$, then $\lambda_k^{\star} = 0$, and $\bm W_k(0) = \bm A_k^H(\bm A_k \bm A_k^H)^{-1}$ is the minimum-Frobenius-norm optimal solution. Otherwise, i.e., $\|\bm W_k(0)\|_F^2 > P_k$, the power constraint is active, and the unique $\lambda_k^{\star} > 0$ satisfies $\|\bm W_k(\lambda_k^\star)\|_F^2 = P_k$. Specifically, since $\|\bm W_k(\lambda_k)\|_F^2$ is strictly decreasing in $\lambda_k$, $\lambda_k^{\star}$ can be efficiently obtained via bisection over $\left[0, \sqrt{\frac{\|\bm A_k\|_F^2}{P_k}}\right]$. Finally, the optimal precoding matrix is given by
\begin{equation}\label{eq:opt_Wk_star2}
    \bm W_k^\star = \bm W_k(\lambda_k^\star).
\end{equation}

\subsection{Optimization of $\bm F$}
For fixed $\{\bm W_k\}$, and after dropping the terms independent of $\bm F$ from the MSE, problem \eqref{eq:prob_MiLAC_simplification} reduces to
\begin{subequations}\label{eq:opt_F}
    \begin{align}
        \min_{\bm F} &~ f(\bm F) \triangleq {\rm Tr} \left( \bm F^H \bm F \bm B \right) - 2 \Re\left\{ {\rm Tr} \left( \bm F \bm C \right)\right\} \label{eq:opt_F_obj} \\[1ex]
        {\rm s.t.} &~ \left\|\bm F\right\|_2 \le 1 \label{eq:opt_F_const}
    \end{align}
\end{subequations}
where $\bm B = \sum_{k = 1}^K \frac{1}{16} \bm H_k \bm W_k \bm W_k^H \bm H_k^H + \frac{1}{4} \sigma_n^2 \bm I_M$, and $\bm C = \sum_{k = 1}^K \frac{1}{4} \bm H_k \bm W_k$. Since $\bm B \succeq \frac{1}{4} \sigma_n^2 \bm I_M \succ \bm 0$, $f(\bm F)$ is strongly convex in $\bm F$. Moreover, the spectral-norm ball in \eqref{eq:opt_F_const} is closed and convex. Therefore, \eqref{eq:opt_F} is a convex optimization problem, which we solve using PGD as detailed below.

The gradient of $f(\bm F)$ with respect to $\bm F^{*}$ is
\begin{equation}\label{eq:opt_F_PGD_gra}
\bm G(\bm F) \triangleq \frac{\partial f(\bm F)}{\partial\bm F^*} = \bm F\bm B-\bm C^H.
\end{equation}
Define the feasible set as $\mathcal F \triangleq \{\bm F \in \mathbb C^{L \times M} | \left\|\bm F \right\|_2 \le 1 \}$. Accordingly, the PGD update at the $r$-th inner iteration is
\begin{equation}\label{eq:opt_F_PGD_F}
    \bm F^{(r)} = \operatorname{Proj}_{\mathcal F} \left( \bm F^{(r - 1)} - \eta \bm G \left(\bm F^{(r - 1)}\right) \right),
\end{equation}
where $\eta > 0$ is the step size and $\operatorname{Proj}_{\mathcal F}(\cdot)$ denotes the Euclidean projection under the Frobenius norm onto $\mathcal F$. The projection onto $\mathcal F$ is obtained by clipping the singular values at one. Specifically, given an arbitrary matrix $\bm Q \in \mathbb C^{L \times M}$, let ${\bm Q} = \bm U \operatorname{diag}\{\sigma_1, \ldots, \sigma_L\} \bm V^H$ denote its compact SVD. Then
\begin{equation*}
\operatorname{Proj}_{\mathcal F}(\bm Q) = \bm U \operatorname{diag}\left(\min\{\sigma_1, 1\}, \ldots, \min\{\sigma_L, 1\} \right) \bm V^H.
\end{equation*}
To determine the step size, we note that $\bm G(\bm F)$ is Lipschitz continuous. In particular, for any $\bm F_1$ and $\bm F_2$, we have
\begin{align*}
    \left\|\bm G(\bm F_1) - \bm G(\bm F_2)\right\|_F
    = & \left\|\left(\bm F_1 - \bm F_2 \right) \bm B \right\|_F \\[2ex]
    \le & \left\|\bm F_1 - \bm F_2 \right\|_F \left\| \bm B \right\|_2 \\[2ex]
    = & \lambda_{\max} (\bm B) \left\|\bm F_1 - \bm F_2 \right\|_F.
\end{align*}
Thus, a Lipschitz constant of $\bm G (\bm F)$ is $L_G=\lambda_{\max}(\bm B)$. We therefore set $\eta = L_G^{-1}$, with which the PGD iterations converge to the unique global optimum of \eqref{eq:opt_F}.

\subsection{Overall Algorithm}
The proposed PGD-based AO framework for solving \eqref{eq:prob_MiLAC_simplification} is summarized in Algorithm~\ref{Alg:MiLAC_AO_PGD}.

\begin{algorithm}[!htbp]
    \caption{Pseudocode for the PGD-based AO Algorithm}
    \label{Alg:MiLAC_AO_PGD}
    \begin{algorithmic}[1]
        \State \textbf{Input}: $\{\bm H_k\}$, $\{P_k\}$;
        \State Initialize feasible $\{\bm W_k^{(0)}\}$ and $\bm F^{(0)}$; set $t = 0$;
        \While{the AO stopping criterion is not satisfied}
            \State Update $\{\bm W_k^{(t+1)}\}$ for fixed $\bm F^{(t)}$ using \eqref{eq:opt_Wk_star2};
            \State Construct $\bm B$ and $\bm C$ using $\{\bm W_k^{(t+1)}\}$;
            \State Initialize the PGD iterations with $\bm F^{(t)}$;
            \State Obtain $\bm F^{(t+1)}$ by iterating \eqref{eq:opt_F_PGD_F} until convergence;
            \State Set $t \leftarrow t+1$;
        \EndWhile
        \State Set $\bm W_k^{\star} = \bm W_k^{(t)}$, $\forall k \in \mathcal{K}$, and $\bm F^\star = \bm F^{(t)}$;
        \State Recover $\bm\Theta^\star$ and $\bm Y^\star$ using \eqref{eq:Theta_recovery} and \eqref{eq:Y_recovery};
        \State \textbf{Output}: $\{\bm W_k^\star\}$, $\bm F^\star$, $\bm\Theta^\star$, and $\bm Y^\star$.
    \end{algorithmic}
\end{algorithm}

We next establish the convergence of the objective values generated by Algorithm~\ref{Alg:MiLAC_AO_PGD}. Let $\mathcal J^{(t)} \triangleq {\rm MSE} \left(\bm F^{(t)}, \{\bm W_k^{(t)}\}\right)$ denote the objective value at the end of the $t$-th AO iteration. During the $(t+1)$-th iteration, we have
\begin{subequations}
    \begin{align}
        \mathcal J^{(t)} & \overset{(a)}{\ge} {\rm MSE} \left(\bm F^{(t)}, \{\bm W_k^{(t + 1)}\} \right) \\[1ex]
        & \overset{(b)}{\ge} {\rm MSE} \left(\bm F^{(t + 1)}, \{\bm W_k^{(t + 1)}\} \right) = \mathcal J^{(t + 1)},
    \end{align}
\end{subequations}
where $(a)$ holds because $\{\bm W_k^{(t+1)}\}$ globally solves \eqref{eq:opt_W} and $(b)$ holds because the PGD iterations do not increase the objective of \eqref{eq:opt_F}. Since the MSE is lower bounded by zero, the sequence $\{\mathcal J^{(t)}\}$ is monotonically non-increasing and therefore converges to a finite limit. Moreover, as the original joint problem is nonconvex, convergence of the objective values does not imply global optimality of the obtained solution.

\subsection{Recovery of the MiLAC Circuit Parameters}
After the AO algorithm converges and yields the optimized scaled aggregation matrix $\bm F^\star$, the corresponding MiLAC circuit parameters can be recovered through a standard network-synthesis procedure \cite{MatteoNerini2025MiLACI-TSP,MatteoNerini2025MiLACII-TSP}. By Lemma~\ref{lem:unitary_completion}, $\bm F^\star$ admits a symmetric and unitary completion of the form
\begin{equation}\label{eq:Theta_recovery}
    \bm\Theta^\star
    =
    \begin{bmatrix}
        \bm\Theta_{11}^\star & (\bm F^\star)^T\\
        \bm F^\star & \bm\Theta_{22}^\star
    \end{bmatrix},
\end{equation}
where $\bm\Theta_{11}^\star$ and $\bm\Theta_{22}^\star$ can be constructed using the SVD-based unitary-completion procedure in \cite[Proposition 1]{ZheyuWu2026MiLACMISO-arXiv}. For a completion satisfying $\det(\bm I_{M+L}+\bm\Theta^\star)\neq0$, the corresponding finite admittance matrix is obtained through the inverse Cayley transform as
\begin{equation}\label{eq:Y_recovery}
    \bm Y^\star = \frac{1}{Z_0} \left(\bm I_{M+L}-\bm\Theta^\star\right) \left(\bm I_{M+L} + \bm\Theta^\star\right)^{-1}.
\end{equation}
Let $\bar{Y}_{p, p}^\star$ denote the admittance connecting the $p$-th port of the MiLAC to ground, and let $\bar{Y}_{p, q}^\star$ denote the admittance connecting the $p$-th and $q$-th ports. According to \cite{MatteoNerini2025MiLACI-TSP} and \cite{MatteoNerini2025MiLACII-TSP}, the tunable admittance components can be recovered as
\begin{equation}\label{eq:admittance_recovery}
    \bar{Y}_{p, q}^{\star} = \begin{cases}
        - [\bm Y^{\star}]_{p, q} & p \neq q \\[1ex]
        \sum_{j = 1}^{M + L} [\bm Y^{\star}]_{j, q} & p = q
    \end{cases} ~,
\end{equation}
for $p,~q = 1, \cdots, M + L$. 

Since $\bm\Theta^\star$ is symmetric and unitary, $\bm Y^\star$ is symmetric and skew-Hermitian, and hence purely imaginary. Therefore, the recovered circuit satisfies the reciprocity and losslessness requirements. It is worth noting that this recovery is performed only once after AO convergence and does not affect the iterative beamforming optimization.

\subsection{Complexity Analysis} \label{subsec:Complexity}
The computational complexity of Algorithm~\ref{Alg:MiLAC_AO_PGD} mainly arises from updating $\{\bm W_k\}$ and $\bm F$ in each AO iteration. Specifically, updating $\{\bm W_k\}$ requires $\mathcal{O}\left( K \left( L M N + N L^2 + J_1 (L^3 + N L^2 + N L) \right) \right)$ operations, where $J_1$ denotes the number of bisection iterations. Constructing $\bm B$ and $\bm C$ requires $\mathcal{O}\left( K (M N L + M^2 L) \right)$ operations, while computing $\lambda_{\max}(\bm B)$ for the PGD step size requires $\mathcal O(M^3)$ operations. Each PGD iteration has complexity $\mathcal O(L M^2 + L^2 M + L M)$, including the gradient computation and the projection onto the spectral-norm ball. Therefore, the overall complexity is $\mathcal O \big(I_1 (K L M N + K (M^2 L + N L^2) + K J_1 (L^3 + N L^2 + N L) + M^3 + J_2 (L M^2 + L^2 M + L M)) \big)$, where $J_2$ and $I_1$ denote the numbers of PGD and AO iterations, respectively. The one-time recovery of $\bm \Theta^\star$ and $\bm Y^\star$ is excluded from the above iterative complexity.

\section{Numerical Results}\label{Sec:NR}
In this section, we verify the convergence of the proposed PGD-based AO algorithm and evaluate the aggregation performance of the proposed MiLAC-aided beamforming design. To enable fair comparisons across different numbers of data streams, we adopt the average MSE per stream, defined as ${\rm MSE}_{\rm stream} = {\rm MSE} / L$.

Following \cite{MatteoNerini2025MiLACII-TSP}, the entries of $\{\bm H_k\}$ are independently generated according to $\mathcal{CN}(0,1)$. The normalized source-signal power budget of each ED is set to $P_k = 10$, $\forall k \in \mathcal K$, and the nominal source SNR is defined as $\text{SNR} \triangleq {P_k}/{\sigma_n^2}$. All results are averaged over 1000 independent channel realizations.

For comparison, we consider one algorithmic benchmark and two beamforming architecture benchmarks: 1) \textbf{AO-SDP}, which replaces the PGD update of $\bm F$ with the SDP formulation in Appendix~\ref{App:A}; 2) \textbf{Digital} \cite{GuangxuZhu2019MIMO-IOTJ}, which employs fully digital beamforming with $M$ RF chains at the AP; and 3) \textbf{Hybrid} \cite{XiongfeiZhai2021Hybrid-TCOM}, which employs phase-shifter-based hybrid beamforming with $L$ RF chains at the AP.

\begin{figure}[htbp!]
\vskip2pt
\centering
\includegraphics[width = 7.8cm]{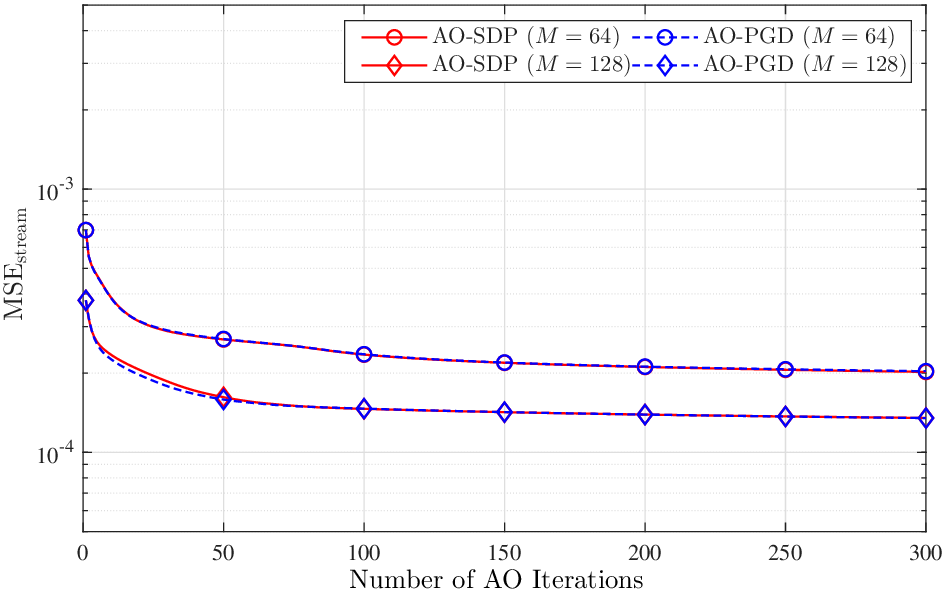}
\caption{Per-stream MSE versus the number of AO iterations, where $K = 20$, $N = 4$, $L = 2$, $\text{SNR} = 10$~dB, and $M \in \{64, 128\}$.} \label{fig:MSE_iter}
\end{figure}

Fig.~\ref{fig:MSE_iter} verifies the convergence of the proposed AO-PGD algorithm and evaluates the accuracy of its PGD update against the AO-SDP benchmark. The objective values decrease steadily and eventually stabilize, consistent with the convergence analysis. More importantly, AO-PGD and AO-SDP exhibit nearly overlapping convergence curves and attain essentially identical objective values, providing numerical evidence that PGD converges to the globally optimal solution of the convex $\bm F$-subproblem. Therefore, AO-SDP is used only to validate the PGD update, while all subsequent MiLAC results are obtained using AO-PGD.

\begin{figure}[htbp!]
\centering
\includegraphics[width = 7.8cm]{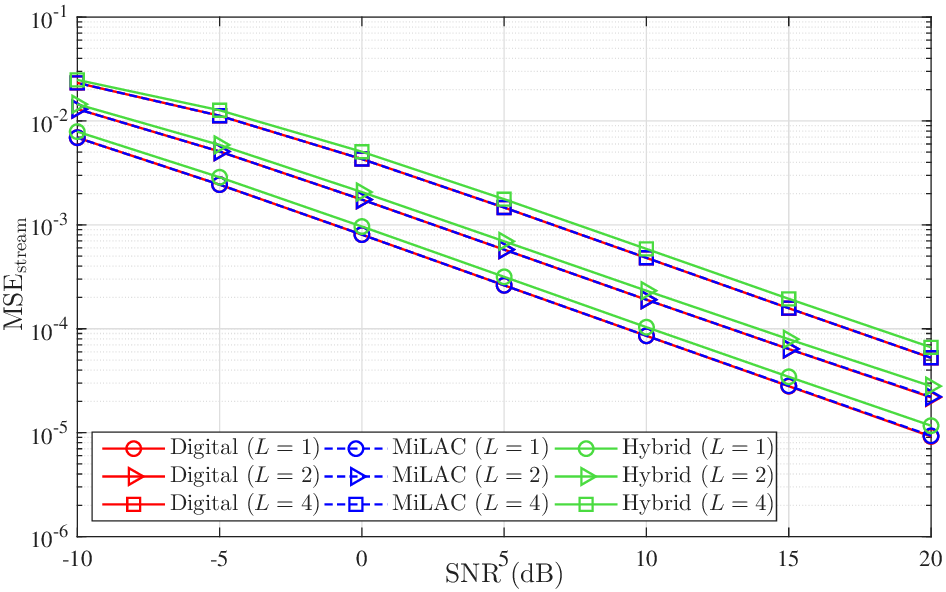}
\caption{Per-stream MSE versus the nominal source SNR, where $K = 20$, $N = 4$, $M = 64$, and $L \in \{1, 2, 4\}$.} \label{fig:MSE_SNR}
\end{figure}

Fig.~\ref{fig:MSE_SNR} compares the aggregation performance of the three beamforming architectures at different nominal source SNRs. MiLAC-aided beamforming closely approaches fully digital beamforming while requiring only $L$, rather than $M$, RF chains at the AP. It also consistently outperforms phase-shifter-based hybrid beamforming under the same RF-chain budget, demonstrating the benefit of the more flexible analog-domain transformations enabled by the MiLAC. The performance degradation with increasing $L$ further indicates that multi-stream aggregation imposes more stringent signal-alignment requirements.

\begin{figure}[htbp!]
\centering
\includegraphics[width = 7.8cm]{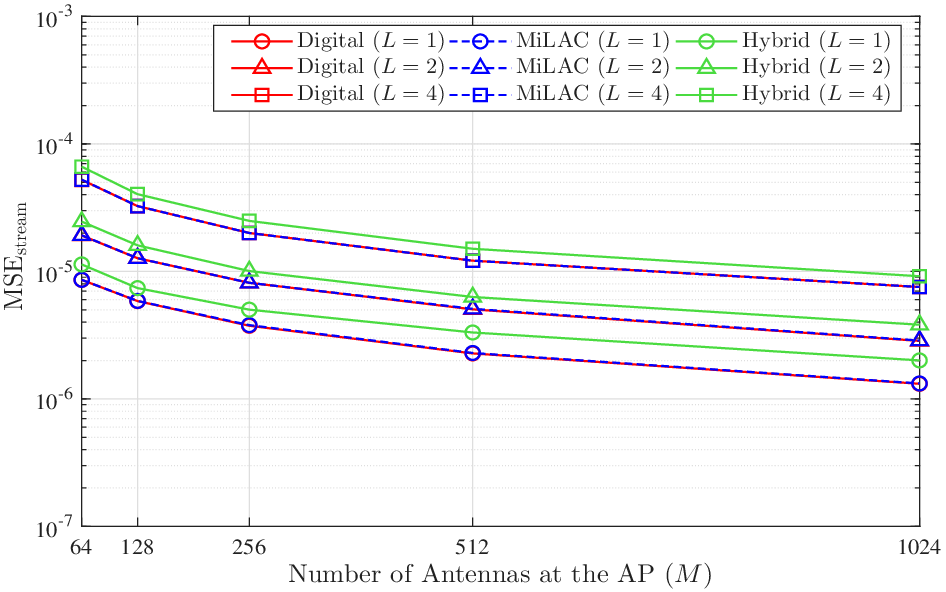}
\caption{Per-stream MSE versus the number of AP antennas $M$, where $K = 20$, $N = 4$, $\text{SNR} = 20$~dB, and $L \in \{1, 2, 4\}$.} \label{fig:MSE_M}
\end{figure}

Fig.~\ref{fig:MSE_M} evaluates the scalability of the considered beamforming architectures with the AP array size. MiLAC-aided beamforming remains close to fully digital beamforming over the entire range of $M$, while its RF-chain requirement remains fixed at $L$ rather than increasing with $M$. Thus, its RF-chain savings become increasingly pronounced as the antenna array scales up. Its consistent advantage over hybrid beamforming further confirms that the flexible analog-domain transformations enabled by the MiLAC are particularly effective for multi-stream AirComp.

\section{Conclusions}\label{Sec-CN}
In this paper, we investigated MiLAC-aided beamforming for MIMO AirComp systems. By exploiting an equivalent spectral-norm characterization of lossless and reciprocal MiLACs, we formulated a joint MSE minimization problem over the ED-side digital precoding matrices and the AP-side MiLAC aggregation matrix. To address the resulting nonconvex problem, we developed an AO algorithm in which the digital precoding matrices are optimally updated using the KKT conditions and one-dimensional bisection, while the MiLAC aggregation matrix subproblem is solved to its global optimum using PGD. The optimized aggregation matrix can subsequently be mapped to the physical MiLAC circuit parameters through a standard network-synthesis procedure. Numerical results verified the convergence of the proposed algorithm and demonstrated that MiLAC-aided beamforming approaches the MSE performance of fully digital beamforming with substantially fewer RF chains at the AP, while outperforming phase-shifter-based hybrid beamforming under the same RF-chain budget. Furthermore, although the proposed architecture substantially reduces the number of RF chains at the AP, the fully connected MiLAC requires a quadratic number of tunable admittances. Reduced-connectivity MiLAC architectures and practical component constraints will be investigated in future work.

\begin{appendices}
\section{SDP Formulation of the $\bm F$-Subproblem}\label{App:A}
In this appendix, we derive an equivalent SDP formulation of the $\bm F$-subproblem \eqref{eq:opt_F}. By introducing a real-valued auxiliary variable $u$, problem \eqref{eq:opt_F} can be equivalently written in the following epigraph form:
\begin{subequations}\label{eq:opt_F_SDP_prob}
    \begin{align}
        \min_{\bm F,u}~~&
        u - 2\Re\left\{{\rm Tr}(\bm F\bm C)\right\}
        \label{eq:opt_F_SDP_obj}\\
        {\rm s.t.}~~&
        {\rm Tr}(\bm F^H\bm F\bm B)\le u,
        \label{eq:opt_F_SDP_const_trace}\\
        &\|\bm F\|_2\le 1.
        \label{eq:opt_F_SDP_const_spectrum}
    \end{align}
\end{subequations}

Let $\bm B^{1/2}$ denote the Hermitian square root of $\bm B$. The quadratic term in \eqref{eq:opt_F_SDP_const_trace} satisfies
\begin{equation}
    {\rm Tr}(\bm F^H\bm F\bm B)
    =
    \|\bm F\bm B^{1/2}\|_F^2
    =
    \left\|\operatorname{vec}(\bm F\bm B^{1/2})\right\|_2^2.
\end{equation}
Therefore, by the Schur complement theorem, constraint \eqref{eq:opt_F_SDP_const_trace} is equivalent to the linear matrix inequality (LMI)
\begin{equation}\label{eq:opt_F_SDP_LMI_trace}
    \begin{bmatrix}
        u & \operatorname{vec}(\bm F\bm B^{1/2})^H\\
        \operatorname{vec}(\bm F\bm B^{1/2}) & \bm I_{LM}
    \end{bmatrix}
    \succeq \bm 0.
\end{equation}
Similarly, the spectral-norm constraint \eqref{eq:opt_F_SDP_const_spectrum} is equivalent to
\begin{equation}\label{eq:opt_F_SDP_LMI_spectrum}
    \begin{bmatrix}
        \bm I_M & \bm F^H\\
        \bm F & \bm I_L
    \end{bmatrix}
    \succeq \bm 0.
\end{equation}
Consequently, \eqref{eq:opt_F_SDP_prob} has a linear objective and two LMI constraints, and is therefore an SDP that can be solved to global optimality using CVX with an SDP solver.
\end{appendices}

\end{document}